 \documentclass[11pt]{article}
\usepackage[paper=a4paper,textwidth=162mm, textheight=238mm]{geometry}
\usepackage{amsmath, amsthm}
\usepackage{amssymb}
\usepackage{amsfonts}
\usepackage{dsfont}
  \usepackage{epsfig} 
\usepackage{graphicx}  \usepackage{epstopdf}
 \usepackage[colorlinks=true]{hyperref}
\hypersetup{urlcolor=blue, citecolor=red}

\newtheorem{theorem}{Theorem}[section]
\newtheorem{corollary}{Corollary}

\newtheorem{lemma}[theorem]{Lemma}
\newtheorem{proposition}{Proposition}

\theoremstyle{definition}
\newtheorem{definition}[theorem]{Definition}
\newtheorem{remark}{Remark}

\def \A{\mathbb{A}}
\def \E{\mathbb{E}}
\def \F{\mathbb{F}}
\def \R{\mathbb{R}}
\def \N{\mathbb{N}}
\def \P{\mathbb{P}}
\def \Q{\mathbb{Q}}
\def \Fc{\mathcal{F}}
\def \Pc{\mathcal{P}}
\def \Tc{\mathcal{T}}
\def \demi{\frac{1}{2}}
\usepackage{dsfont}
\def \ind#1{\mathds{1}_{\{#1\}}}

\def \x{\times}
\def \And{\;\mbox{ and }\;}

\def \cadlag{c\`adl\`ag }

\usepackage{authblk}
\title{Investment under uncertainty, competition and regulation}
\author[1]{Adrien Nguyen Huu\thanks{Corresponding author: adrien.nguyenhuu@gmail.com}}
\affil[1]{{\small{ \it IMPA, Estrada Dona Castorina 110, Rio de Janeiro 22460-320, Brasil}}}

\begin{document}
\maketitle

\begin{abstract}
We investigate
a randomization procedure
undertaken in real option games
which can serve as a basic model
of regulation
in a duopoly model of
preemptive investment.
We recall the rigorous
framework of M. Grasselli, 
V. Lecl\`{e}re and M. Ludkovsky 
(\textit{Priority Option: the value of being a leader}, 
International Journal of Theoretical and Applied Finance, \textbf{16}, 2013),
and extend it to a random
regulator.
This model generalizes
and unifies the different
competitive frameworks proposed
in the literature,
and creates a new one
similar to a Stackelberg leadership.
We fully characterize
strategic interactions
in the several situations following
from the parametrization of the regulator.
Finally, we study the effect of 
the coordination game and uncertainty of outcome
when agents are risk-averse,
providing new intuitions for
the standard case.
\end{abstract}

\section{Introduction}

A significant progress on the net present valuation
method has been made
by real option theory
for new investment valuation.
The latter uses
recent methods from stochastic finance
to price uncertainty
and competition in complex investment situations.
Real option games
are the part of this theory
dedicated to the competitive dimension.
Following the seminal work of Smets \cite{smets1993},
they correspond to
an extensively studied situation,
where two or more economic agents
face with time
a common project to invest into,
and where there might be
an advantage of being 
a leader
(\emph{preemptive} game) 
or a follower
(\emph{attrition} game).
Grenadier \cite{grenadier1996},
Paxson and Pinto \cite{paxson2005}
or Weeds \cite{weeds2002} among many others
develop nice examples of this type of model.
On these problems
and the related literature,
Azevedo and Paxson
\cite{azavedo2010} or
Chevalier-Roignant \textit{et al.} \cite{chevalier2011} provide
comprehensive reviews.
Real option games are especially dedicated
to new investment opportunities
following R\&D developments,
like technological products,
a new drug or even
real estate projects.
The later investment opportunities
indeed add competition risk to 
the uncertainty of future revenues
in the strategic behavior of investors.
The regulatory dimension
is also an important part
of such investment projects,
and the present paper attempts
to contribute in this direction.

The economic situation can be the following.
We consider an investment opportunity
on a new market for an homogeneous good,
available for two agents labeled one and two.
This investment opportunity is
not constrained in time,
but the corresponding market 
is regulated by an
outside agent.
Both agents have the same model
for the project's revenues, 
and both have also access to
a financial market with one risky asset (a market portfolio)
and a risk-free asset (a bank account).
The situation is a non-cooperative
duopoly of preemptive timing game for
a single investment oppotunity.

A simple approach to regulation
is to consider that
when an agent wants to
take advantage of such an opportunity,
he must satisfy some non-financial criteria
which are scrutinized by the regulator.
Even though a model
of regulatory decision
is at the heart of this matter,
we will take the most simple and naive approach.
We will just assume that the regulator
can accept or refuse an agent to
proceed into the investment project,
and that this decision is taken randomly.
This extremely simple model
finds its root in a widespread
convenient assumption.
In the standard real option game
representing a Stackelberg duopoly,
the settlement of leadership and followership
is decided via the flipping of an even coin,
independent of the model.
The procedure is used
for example in Weeds \cite{weeds2002},
Tsekeros \cite{tsekrekos2003}
or Paxson and Pinto \cite{paxson2005}
who all refer to 
Grenadier \cite{grenadier1996}
for the justification, 
which appears in a footnote
for a real estate investment project:
\begin{quotation}
A potential rationale for this assumption 
is that development requires approval from 
the local government. 
Approval may depend on 
who is first in line or 
arbitrary considerations.
\end{quotation}
This justification
opens the door
to many alterations
of this assumption,
inspired from other similar
economic situations:
the aforementioned considerations 
might lead to favor one competitor on the other or
lead to a different outcome.

Let us take a brief moment
to describe one of those economic situations we have in mind.
Assume that two economic agents
are running for
the investment in a project
with possibility 
of simultaneous investment,
as in Grasselli \textit{et al.} \cite{grasselli2013}.
In practice,
even if they are 
accurately described as symmetrical,
they would never act at 
the exact same time,
providing that instantaneous
action in a continuous time
model is just an idealized
situation.
Instead,
they
show their intention
to invest to a third party
(a regulator, a public institution)
at approximatively the same time.
An answer to a call for tenders
is a typical example of such interactions.
Then the arbitrator
is invoked to 
judge the validity of these intentions.
For example,
he can evaluate
which agent is the most suitable
to be granted the project as a leader
regarding qualitative criteria.
This situation might be
illustrated in particular
where environmental or
health exigences are in line.
When simultaneous investment is impossible,
the real estate market example
of Grenadier \cite{grenadier1996} can also be
cited again bearing in mind
that, in addition to safety constraints,
an aesthetic or confidence
indicator can intervene in 
the decision of a market regulator.
Because these criteria can be numerous,
an approach via randomness
of the decision is a decent first approximation.
In the extremal case,
the arbitrator can be biaised and
show his preference for on agent.
In general, it leads to an asymmetry
in the chance to be elected leader
and an unhedgeable risk,
and perfectly informed agents
should take into account
this fact
into their decision to invest or
to defer.
Those are some of the
situations
the introduction of a
random regulator
can shed some light
on.

Our model is presented
in a very particular way by using a normal form game
in an extended natural virtual time,
as initially proposed by Fudenberg and Tirole \cite{fudenberg1985},
and followed by Thijssen and \textit{al.} \cite{thijssen2002}
or Grasselli and \textit{al.} \cite{grasselli2013}.
This model bears a very specific interpretation,
but appears to
encompass in a continuous manner
the Cournot competition setting
of Smets \cite{smets1993} and
the Stackelberg competition of
Grenadier \cite{grenadier1996}.
The introduction of a regulator also
renders the market incomplete.
We then follow two classical approaches,
namely the risk-neutral and the risk-averse cases respectively.
The latter case actually brings up
new insights on the strategic interactions
of players in the Cournot and Stackelberg setttings.
Against the existing literature,
the present contribution situates itself
as a contribution to mathematical clarifications
of the standard real option game, with new angles.

Let us present how
the remaining of the paper proceeds.
Section \ref{sec:game}
introduces the standard model
and its extension
to the random regulatory framework.
Section \ref{sec:nash}
provides the study
of the model and optimal strategies
in the general case.
In Section \ref{sec:singular},
we present how the proposed
model encompasses the usual types
of competition, and propose a new
asymmetrical situation related
to the Stackelberg competition framework.
Section \ref{sec:aversion}
introduces risk-aversion in agents evaluation
to study the effect of
a random regulator and the uncertain
outcome of a coordination game
on equilibrium strategies of
two opponents.
Section \ref{sec:conclusion}
finishes on critism and possible extentions.

\section{The model}
\label{sec:game}

\subsection{The investment opportunity}
\label{sec:market}

We assume for the moment that 
the regulator does not intervene.
The framework is thus the standard one
that one shall find for example in Grasselli \textit{et al.} \cite{grasselli2013}.
Notations and results of this section
follow from the latter.
The later research will then be 
cited repeatedly in the present article.

We consider
a stochastic basis $(\Omega, \Fc, \F, \P)$
where $\F:=(\Fc_t)_{t\ge0}$
is the filtration
of a one-dimensional Brownian motion
$(W_t)_{t\ge0}$, i.e., 
$\Fc_t := \sigma\{W_s, \ 0\le s\le t\}$.
The project delivers for
each involved agent
a random continuous stream of cash-flows
$(D_{Q(t)}Y_t)_{t\ge0}$.
Here,
$Y_t$ is the stochastic profit
per unit sold and
$D_{Q(t)}$ is the
quantity of sold units per agent
actively involved on the market,
when $Q(t)$ agents are actively participating at time $t$.
The increments of $(Q(t))_{t\ge 0}$
inform on the timing
decision of agents,
and considering a perfect information setting,
we assume that $(Q(t))_{t\ge 0}$
is a $\F$-adapted right-continuous process:
agents are informed of competitors entry in the market,
and expected future cash-flows can be computed.
It is natural to assume that
being alone in the market is
better than sharing it with a competitor:
\begin{equation}
\label{eq:demand}
0=:D_0<D_2<D_1 \;,
\end{equation}
but we also assume these quantities
to be known and constant.
The process $(Y_t)_{t\ge0}$ is 
$\F$-adapted,
non negative and
continuous
with dynamics given by
\begin{equation}
\label{eq:profit}
dY_t = Y_t( \nu dt + \eta dW_t), \quad t\ge 0\; ,
\end{equation}
with $(\nu,\eta)\in \R\x\R_+^*$.
We assume that $(Y_t)_{t\ge0}$
is perfectly correlated to
a liquid traded asset whose price dynamics is given 
by
\begin{equation}
\label{eq:price}
dP_t = P_t (\mu dt + \sigma dW_t) = P_t (rdt + \sigma dW^{\Q}_t)
\end{equation}
where 
$(\mu,\sigma)\in \R\x\R_+^*$, 
$r$ is the constant interest rate
of the risk-free bank account
available for both agents, 
and $W^{\Q}_t = W_t + \lambda t$
is a Brownian motion under the unique
risk-neutral measure $\Q\sim \P$
of the arbitrage-free market.
The variable $\lambda:=(\mu-r)/\sigma$ in
\ref{eq:price} is the Sharpe ratio.
The present financial setting is
thus the standard Black-Scholes-Merton
model \cite{blackscholes1973} of
a complete perfect financial market.

\begin{remark}
 \label{rem:units}
 The above setting defines $Y_t$ as a random
 unitary profit and $D$ as a fixed quantity, 
 whereas Grenadier \cite{grenadier2000} 
 and Grasselli \textit{et al.}\cite{grasselli2013}
 labelled $Y_t$ and $D$ as the demand
 and the inverse demand curve respectively.
 This model choice has no incidence in complete market
 and Sections \ref{sec:follower} and \ref{sec:leader} herefater.
 The choice will be discussed in Remark \ref{rem:incomplete}
 regarding the introduction of the regulator in Section \ref{sec:regulator}.
\end{remark}

\subsection{The follower's problem}
\label{sec:follower}

In this setting,
since
$D_{Q(t)}$ takes
known values,
the future stream of cash-flows 
can be evaluated under the
risk-neutral measure under which 
\begin{equation}
\label{eq:risk neutral Y}
dY_t = Y_t ((\nu-\eta\lambda) dt + \eta dW^{\Q}) \; .
\end{equation}
Assume that one of the two agents,
say agent one,
desires to invest at time $t$
when $Y_t=y$.
If $Q(t)=1$, 
then the available market
for agent one is $D_2$.
The risk-neutral expectation of
the project's discounted cash flows are then given by
\begin{equation}
\label{eq:VF}
V^F(t,y) 
:= \E^{\Q} \left[ \int_t^{\infty}e^{-r(s-t)}D_2 Y_s ds \right]
= \frac{D_2 y}{\eta \lambda - (\nu - r)} 
= \frac{D_2 y}{\delta}
\end{equation}
with $\delta:= \eta \lambda - (\nu - r)$.
We assume from now on $\delta>0$.
Now to price the
investment option of a follower,
we recall that agent one
can wait to invest as long as he wants.
We also recall the sunk cost 
$K$ at time $\tau$
he invests.
In the financial literature,
this is interpreted as
a Perpetual American call option
of payoff $(D_2Y_{\tau}/\delta - K)^+$.
The value function of this
option is given by
\begin{equation}
\label{eq:follower}
F(t,y):=
\sup\limits_{\tau\in \Tc_t}\E^{\Q}\left[
	e^{-r(\tau-t) }\left(
		\frac{D_2 Y_{\tau}}{\delta}-K\right)^+ \ind{\tau<+\infty} | \Fc_t
	\right]
\end{equation}
where $\Tc_t$ denotes the collection of
all $\F$-stopping times with values in
$[t,\infty]$.
The solution to \ref{eq:follower}
is well-known
in the literature, see Huang and Li \cite{huang1991} and Grenadier \cite{grenadier1996}.
A formal recent proof can be found in
Grasselli \textit{et al.} \cite{grasselli2013}.

\begin{proposition}[Prop.1, \cite{grasselli2013}]
\label{prop:follower_complete}
The solution to \ref{eq:follower} is given by
\begin{equation}
\label{eq:follower_sol}
F(y) = \left\{
	\begin{array}{ll}
	\frac{K}{\beta - 1}\left(\frac{y}{Y_F}\right)^{\beta} 
	& \text{ if } y \le Y_F, \\
	\frac{D_2 y}{\delta}-K 
	& \text{ if } y > Y_F, \\	
	\end{array}
	\right.
\end{equation}
with a threshold $Y_F$ given by
\begin{equation}
\label{eq:Y_F}
Y_F := \frac{\delta K \beta}{D_2 (\beta-1)}
\end{equation}
and
\begin{equation}
\label{eq:beta}
\beta := \left(\demi - \frac{r-\delta}{\eta^2}\right)
+ \sqrt{\left(\demi - \frac{r-\delta}{\eta^2}\right)^2 + \frac{2r}{\eta^2}}>1 .
\end{equation}
\end{proposition}
The behavior of
the follower is thus quite explicit.
He will defer investment
until the demand reaches at least
the level $Y_F=\beta/(\beta-1)K>K$
which depends on the 
profitability of
the investment opportunity,
the latter being conditioned
by $\delta>0$.
We thus introduce
\begin{equation}
\label{eq:tau follower}
\tau_F := \tau(Y_F) = \inf\{t\ge 0~:~Y_t \ge Y_F \} \; .
\end{equation}

\subsection{The leader's problem}
\label{sec:leader}
Assume now that instead of
having $Q(t)=1$ we have $Q(t)=0$.
Agent one investing at 
time $t$ will receive a stream of cash-flows
associated to the level $D_1$
for some time,
but he expects agent two
to enter the market
when the threshold $Y_F$
is triggered.
After the moment $\tau_F$,
both agents share the market
and agent one receives cash-flows
determined by level $D_2$.
The project value is thus

\begin{align*}
V^L (t,y)&:=\E^{\Q}\left[
	\int_t^{\infty}\hspace{-0.3cm} e^{-r (s-t)}(D_1\ind{s< \tau_F}+ D_2 \ind{s\ge \tau_F}) Y_s^{t,y}ds
\right]\\
 &= \frac{D_1 y}{\delta} - \frac{(D_1-D_2)Y_F}{\delta}\left(\frac{y}{Y_F}\right)^{\beta}
\end{align*}
where detailed computation can be found
in  Grasselli \textit{et al.} \cite{grasselli2013}.
This allows to characterize
the leader's value function $L(t,y)$, i.e.,
the option to invest at time $t$ for
a demand $y$,
as well as
the value of the project $S(t,y)$
in the situation of 
simultaneous investment.

\begin{proposition}[Prop. 2, \cite{grasselli2013}]
\label{prop:leader_complete}
The value function of a leader is given by 
\begin{equation}
\label{eq:leader_sol}
L(y) = \left\{
	\begin{array}{ll}
	\frac{D_1 y}{\delta} - \frac{(D_1-D_2)}{D_2} \frac{K\beta}{\beta-1}\left(\frac{y}{Y_F}\right)^{\beta}
	& \text{ if } y < Y_F,\\
	\frac{D_2 y}{\delta}-K 
	& \text{ if } y \ge Y_F ,\\	
	\end{array}
	\right.
\end{equation}
If both agents invest simultaneously,
we have
\begin{equation}
\label{eq:simultaneous_sol}
S(y):= \frac{D_2 y}{\delta}-K \;.
\end{equation}
\end{proposition}

\begin{remark}
Notice that no
exercise time is involved
as we consider the interest
of exercising immediately,
$Y$ being non-negative.
Notice also that $L$, $F$ and $S$
do not depend on $t$,
since
the problem is stationary.
\end{remark}

The payoff of the investment opportunity is then fully
characterized for any situation in the case
of no regulatory intervention.

\subsection{The regulator}
\label{sec:regulator}

Let us define $\tau_1$
and $\tau_2$
the time at which
agents one and two respectively
express their desire to invest.
In full generality,
$\tau_i$ for $i=1,2$
can depend on a probability of acting
in a game, see the next subsection.
We assume that agents
cannot predict
the decision of
the regulator, 
so that
$\tau_1, \tau_2$ are
$\F$-adapted stopping times.
The regulator only intervenes
at such times.
If at time $\tau_i$ for $i=1,2$,
$Q(\tau_i^-)=0$,
but $Q(\tau_i)=1$
then his decision affects
only agent $i$ who expresses
his desire to be a leader.
If however $Q(\tau_i)=2$,
then $\tau_1=\tau_2$
and the regulator
shall decide if
one or the other agent
is accepted, none is or
both are.
Finally,
if $Q(\tau_i^-) = 1$,
then the regulator
takes his decision
upon the follower's fate.
The regulator decision
thus depends on $\F$.

We introduce a
probability space
$(\Lambda, \Pc(\Lambda),\A)$
where $\Lambda=\{\alpha_0, \alpha_1, \alpha_2, \alpha_S\}$.
We then introduce
the product space
$(\Omega \x \Lambda, \Fc \x \Pc(\Lambda), \P^+)$
and the augmented
filtration
$\F^+:=(\Fc^+_t)_{t\ge 0}$
with $\Fc^+_t:=\sigma \{\Fc_t, \Pc(\Lambda) \}$.
The regulator
is modeled by a
$\F^+$-adapted process.

\begin{definition}
\label{def:regulator}
Fix $t$ and $Y_t=y$. For $i=1,2$, if $j=3-i$ is the index of the opponent
and $\tau_j$ his time of investment, then agent $i$ desiring to invest
at time $t$ receives
\begin{equation}
\label{eq:regulator}
R_i(t,y):=\left\{
\begin{array}{ll}
0 & \text{if }\alpha=\alpha_0\\
L(y)\ind{t\le \tau_j}+F(y)\ind{t>\tau_j} & \text{if }\alpha=\alpha_i\\
F(y)\ind{t=\tau_j} & \text{if }\alpha=\alpha_j\\
L(y)\ind{t< \tau_j}+S(y)\ind{t=\tau_j}+F(y)\ind{t>\tau_j} & \text{if }\alpha=\alpha_S\\
\end{array}
\right. \; .
\end{equation}
\end{definition}

Let us discuss briefly this representation.
According to \ref{eq:regulator},
agent $i$ is accepted if alternative $\alpha_i$ or $\alpha_S$
is picked up, 
and denied if $\alpha_0$ or $\alpha_j$ is picked up.
It is therefore implicit in the model
that time does not
affect the regulator's decision upon acceptability.
However $Q(t)$ affects the position of leader or follower.
Probability $\P^+$ can thus
be given by $\P\x\A$,
and probability
$\A$ given by a quartet
$ \{q_0,q_1, q_2, q_S\}$.
However,
as we will see shortly,
the alternative $\alpha_0$
is irrelevant in what follows due to the way
regulatory intevention is modelled.
We assume $q_0<1$.
Since for the 
unregulated model we assumed that
agents are symmetrical,
the general study of the regulator's
parameters given by $\A$
will be chosen without loss of generality
such that $q_1\ge q_2$.

An additional major change comes into the evaluation of payoffs.
Agents information is reduced to $\F$.
Therefore the final settlement is
not evaluable as in the complete market setting,
and we shall introduce a pricing criterion.
We follow Smets \cite{smets1993}, Grenadier \cite{grenadier1996}
and many others by making the usual assumption
that agents are risk-neutral, i.e., 
they evaluate the payoffs
by taking the expectation of
previously computed values under $\A$.
The incomplete market 
is also commonly handled via utility maximization, 
see Benssoussan \textit{et al.} \cite{bensoussan2010} 
and Grasselli \text{et al.} \cite{grasselli2013}.
We postpone this approach to Section \ref{sec:aversion}.

\begin{remark}
\label{rem:incomplete}
 Expectations of \ref{eq:regulator}
 are made under the
 \emph{minimal entropy martingale measure} 
 $\Q\x \A$
 for this problem,
 meaning that uncertainty of the model
 follows a \emph{semi-complete} market hypothesis:
 if we reduce uncertainty
 to the market information $\F$,
 then the market is complete.
 See Becherer \cite{becherer2003} for details.
 Recalling Remark \ref{rem:units},
 we can thus highligh that 
 payoffs $L(y)$, $F(y)$ or $S(y)$ can
 be perfectectly evaluated.
 Therefore definitions
 of $Y$ and $D$ only matter in the interpretation
 of the model and the present choice of
 probability is not affected by such interpretation.
\end{remark}

Once the outcome of $\Lambda$ is settled, 
the payoff then depends on $Q(t)$, i.e.,
$\tau_j$.
We thus focus now on 
strategic interactions that determine $Q(t)$.

\subsection{Timing strategies}
\label{sec:strategies}
It has been observed
since Fudenberg \& Tirole \cite{fudenberg1985}
that real time $t\ge0$ is not sufficient to
describe strategic possibilities of opponents
in a coordination game.
If $Q(t)=0$ and agents
coordinate by deciding to invest with probabilities $p_i\in(0,1)$
for $i=1,2$,
then a one-round game at time $t$ implies
a probability $(1-p_1)(1-p_2)$ to exit without at least one investor.
However, another coordination
situation appears for the instant just after
and the game shall be repeated with the same parameters.
The problem has been settled
by Fudenberg and Tirole \cite{fudenberg1985}
in the deterministic case,
and recently by 
Thijssen \textit{et al.} \cite{thijssen2002}
for the stochastic setting.
We extend it to the model with regulator 
in the following manner.

It consists
in extending time $t\in \R_+$ to
$(t,k,l)\in \R_+ \x \N^*\x \N^*$.
The filtration $\F$
is augmented via
$\Fc_{t,i,j}= \Fc_{t,k,l}\subseteq \Fc_{t',i,j}$
for any $t<t'$ and any $(i,j)\ne (k,l)$,
and the state process is extended to
$Y_{(t,k,l)}:=Y_t$.
Therefore,
when both agents
desire to invest at the same time $t$,
we extend the time line by freezing 
real time $t$ and
indefinitely repeating
a game on natural time $k$.
Once the issue of the game is settled,
the regulator intervenes on natural time $l=1$.
If no participant is accepted, i.e.,
if $\alpha_0$ is picked up.
the game is replayed for $l=2$,
and so on.

This model needs to be commented.
At first glance,
a natural assumption is to make the
regulator intervene once and for all
for each agent or both at the same time.
If the regulator refuses the entry of
a competitor, one could
assume for example that the game is over
or that a delay is imposed before any other
action of the rejected agent.
The time extention above and the following 
Definition \ref{def:strategy} do 
not obey this rule.
However, it can be related
to a unique intervention of the regulator
by forthcoming Corollary \ref{cor:d and e}
and Proposition \ref{prop:table}, 
see Remark \ref{rem:reduction} at the end of
Section \ref{sec:coordination}.
The motivation for the present model
lies in presentation details.
In Definition \ref{def:regulator},
it allows the generic expression \ref{eq:regulator}
and avoids direct conditioning of $R_i(t,y)$
on the position of agent $i$.
Expression \ref{eq:regulator}
can thus be directly handled and modified to 
represent a more realistic situation
than one implying the systematic
entry of at least one player, 
as it will be shown in the next section.
The present model also allows
for the relevant specific cases of Section \ref{sec:singular_cournot}
and Section \ref{sec:stackelberg leadership}.
By using the framework of 
Fudenberg and Tirole \cite{fudenberg1985},
the present model also shows its limitations
as an idealized situation.

\begin{definition}
\label{def:strategy}
A strategy
for agent $i\in \{1,2\}$
is defined
as a pair of 
$\F$-adapted processes
$(G^i_{(t,k,l)}, p^i_{(t,k,l)})$
taking values in $[0,1]^2$ such that
\begin{enumerate}
	\item[(i)] The process $G^i_{(t,k,l)}$  is of the type
	$G^i_{(t,k,l)}(Y_t) = G^i_{t}(Y_t) 
	= \ind{t\ge \tau(\hat y)}$
	with
	$\tau(\hat y):=\inf\{t\ge 0~:~ Y_t\ge \hat y\}$.
	\item[(ii)] The process $p_i(t,k,l)$ is of the type
	$p_i(t,k,l)=p_i(t)= p^i(Y_t)$.
\end{enumerate}
\end{definition}

The reduced set of strategies is motivated
by several facts.
Since the process $Y$ is Markov,
we can focus
without loss of generality
on Markov sub-game perfect equilibrium strategies.
The process $G^i_{(t,k)}$
is a non-decreasing
\cadlag process,
and refers to the cumulative
probability of agent $i$
exercising before $t$.
Its use is kept when agent $i$
does not exercise immediately the option to invest,
and when exercise depends on
a specific stopping time of the form
$\tau(\hat Y)$,
such as 
the follower strategy.
The process $p^i_{(t,k,l)}$ denotes
the probability of exercising
in a coordinate game at round $k$,
after $l-1$ denials of the regulator, 
when $\alpha=\alpha_0$.
It should be
stationary
and not depend on the
previous rounds of the game since
no additional information is given.
For both processes,
the information is given by $\F$ 
and thus reduces to $Y_t$ at time $t$.
Additional details can be found 
in Thijssen \textit{et al.}
\cite{thijssen2002}.

\section{Optimal behavior and Nash equilibria}
\label{sec:nash}

\subsection{Conditions for a coordination game}
\label{sec:coordination}
A first statement about payoffs can be immediately
provided. A formal proof can be found in 
Grasselli \textit{et al.}~\cite{grasselli2013},
following arguments initially developped in Grenadier \cite{grenadier1996, grenadier2000}.

\begin{proposition}[Prop. 1, \cite{grenadier2000}]
\label{prop:thresholds}
There exists a unique point $Y_L\in (0,Y_F)$
such that
\begin{equation}
\label{eq:threshold_discr}
\left\{
	\begin{array}{ll}
	S(y)<L(y)<F(y) & \text{ for } y<Y_L,\\
	S(y)<L(y)=F(y) & \text{ for } y=Y_L,\\
	S(y)<F(y)<L(y) & \text{ for } Y_L<y<Y_F,\\
	S(y)=F(y)=L(y) & \text{ for } y\ge Y_F\; .
	\end{array}
\right.
\end{equation}
\end{proposition}

Fix $t$ and $Y_t=y$.
In the deregulated situation,
three different cases are thus possible,
depending on the three intervals given by $0<Y_L<Y_F<+\infty$.
It appears that in our framework,
the discrimination also applies.

Consider the following situation.
Assume $Q(t)=0$ and $t=\tau_1<\tau_2$:
agent one
wants to start investing in the project as a leader
and agent two allows it, i.e., $(p_1(t), p_2(t))=(1,0)$.
By looking at \ref{eq:regulator},
agent one receives $L(y)$ with probability
$q_1+q_S$ and $0$ with probability $q_2+q_0$.
However, as noticed for
the coordination game,
if agent one is denied investment at $(t,1,1)$,
he can try at $(t,1,2)$ and so on
until he obtains $L(y)$.
Consequently, if $q_0<1$
agents can do \emph{as if} $\alpha_0$ is never picked up,
see Remark \ref{rem:reduction} below.
The setting is limited
by the continuous time approach
and Proposition \ref{prop:thresholds}
applies as well as in the standard case.
The situation is identical if $\tau_2<\tau_1$.

\begin{corollary}
\label{cor:d and e}
Let $t\ge 0$ and $y>0$. Then for $i=1,2$
\begin{enumerate}
	\item[(d)] if $y<Y_L$, agent $i$ defers action until $\tau(Y_L)$;
	\item[(e)] if $y>Y_F$, agent $i$ exercises immediately, i.e., $\tau_i=t$.
\end{enumerate}
\end{corollary}

\begin{proof}
(d) According to \ref{eq:threshold_discr},
if $\tau_1\ne \tau_2$,
expected payoffs given to
\ref{eq:regulator} verify
$(q_S+q_i)L(y)<(q_S+q_i)F(y)$ and
there is no incentive to act for agent $i$,
$i=1,2$.

(e) Since $S(y)=F(y)=L(y)$ and $\tau_F=t$, 
both agents act with probability
$(p_1(t), p_2(t))=(1,1)$.
Since $q_0<1$, they submit their request
as much as needed and
receive $S(y)$ with probability
$(1-q_0)\sum_{l\in \N}q_0^l = 1$. 
\end{proof}

Corollary \ref{cor:d and e} does not
describe which action is undertaken at $\tau(Y_L)$
or if $y=Y_F$.
We are thus left to study the case where $Y_L\le y \le Y_F$.
This is done by considering a coordination game.
Following the reasoning of the above proof,
definition \ref{eq:regulator} allows
to give the expected payoffs of the game.
Let us define
\begin{equation}
\label{eq:payoff_arbitrator}
(S_1, S_2):=\left( \frac{1}{1-q_0}(q_1 L+q_2 F + q_S S), \frac{1}{1-q_0}(q_2 L+q_1 F + q_S S )\right)\; .
\end{equation}
Notice that $S_i$ are bounded by $\max \{L, F, S\}$ since
$q_1+q_2+q_s = 1-q_0$.

\begin{proposition}
\label{prop:table}
For any $q_0<1$,
agents play the coordination game given by 
Table \ref{tab:1}
if $Y_L<y<Y_F$. Consequently, we can assume $q_0=0$
without loss of generality.

\begin{table}[h!]  
  \centering
  \begin{tabular}{|c|c|c|}
    \cline{2-3}
    \multicolumn{1}{c|}{} & Exercise & Defer \\ \hline
    Exercise    & $(S_1(y), S_2(y))$   & $(L(y), F(y))$  \\ \hline
    Defer   & $(F(y), L(y))$   & Repeat \\ \hline
  \end{tabular}
  \caption{Coordination game at $(t, k, l)$.}
  \label{tab:1}
\end{table}
\end{proposition}

\begin{proof}
If $\alpha =\alpha_0$ after settlement at round $k$ of the game, 
time goes from $(t,k,l)$ to $(t,1,l+1)$ and the game is repeated.
Therefore,
according to \ref{eq:regulator},
the game takes the form of Table \ref{tab:1} for a fixed $l$
and is settled with probability $1-q_0$,
or canceled and repeated with probability $q_0$.
If $(p_1, p_2)$ is a given strategy for the game
at $t$ and $(E_1(p_1,p_2), E_2(p_1, p_2))$ the 
consequent expected payoff for agents one and two
in the game of Table \ref{tab:1},
then the total expected payoff for
agent $i$ at time $t$ is
\begin{equation}
\label{eq:total expected payoff}
E_i(p_1,p_2)(1-q_0)\sum_{l\in \N}q_0^l = E_i(p_1,p_2)\; .
\end{equation}
The game is thus not affected by $q_0<1$, 
which can take any value stricly lower than one.
Without loss of generality, 
we can then take $q_0=0$ and reduce the game to one intervention of the regulator, 
i.e. $l\le 1$.

When $\tau_1\ne \tau_2$,
the probability that the regulator
accepts agent $i$ demand for investment
is $q_i+q_S$, see Corollary \ref{cor:d and e}.
There is a complete equivalence of payoffs and strategies
for quartet $\{q_0, q_1, q_2, q_S\}$ and
quartet $\{0, q_1/(1-q_0), q_2/(1-q_0), q_S/(1-q_0)\}$.
Probability $q_0$ can then be settled to zero.
\end{proof}

\begin{remark}
 \label{rem:reduction}
 Extended time model of Section \ref{sec:strategies} 
 implies strong limitations
 to the issues of the game:
 Corollary \ref{cor:d and e}
 and Proposition \ref{prop:table}
 impose the emergence of at least one player
 at $\tau_1$ and $\tau_2$.
 The confrontation has thus three outcomes
 and can bear the following interpretation.
 In the present model of competition,
 a regulator intervenes only if
 the two agents attemps to enter the 
 investment opportunity at the same time.
 He then decides among three alternatives:
 letting one, or the other, or both agents
 enter the opportunity.
 The decision of the regulator is definitive
 and it is thus possible
 to connect this procedure to 
 the original one of Grenadier \cite{grenadier1996}
 mentionned in the introduction, 
 see also Section \ref{sec:singular_cournot}.
\end{remark}

According to Remark \ref{rem:reduction}, we settle $q_0=0$ from now on.
We now turn to the solution to Table \ref{tab:1}.

\subsection{Solution in the regular case}

We reduce the analysis here
to the case where
$0<q_2\le q_1 < 1-q_2$.
We can now assume
that 
\begin{equation}
\label{eq:strategy mixte}
\max(p_1, p_2)>0 \; .
\end{equation}
We now introduce
$p_0(y):= (L(y)-F(y))/(L(y)-S(y))$ and two functions
\begin{equation}
\label{eq:V1}
P_i(y):= \frac{p_0(y)}{q_i p_0(y)+q_S}=\frac{L(y)-F(y)}{L(y)-S_j(y)}
, \quad i\ne j \in\{1,2\}^2 \; .
\end{equation}
The values of $P_i$
strongly discriminates the issue of the game.
Since $q_1\ge q_2$, if $Y_L\le y \le Y_F$,
then $S_1(y)\ge S_2(y)$ according to \ref{eq:threshold_discr}
on that interval and $P_2(y)\ge P_1(y)$.

\begin{lemma}
\label{prop:concave_V}
The functions $P_2$ and $P_1$ are increasing on $[Y_L, Y_F]$.
\end{lemma}

\begin{proof}
By taking $d_1(y):=L(y)- F(y)$ and
$d_2(y):=S(y)-F(y)$,
we get

$$
P_i(y) = \frac{1}{q_i} \left[ \frac{d_1(y)}{d_1(y) + \gamma_i (d_1(y)-d_2(y))} \right]
$$
and

$$
P'_i(y) = \frac{1}{q_i} \left[ \frac{\gamma_i(d_1 (y) d'_2 (y) - d_2(y)d'_1(y))}{(d_1(y) + \gamma_i (d_1(y)-d_2(y)))^2} \right]
$$
where 
$\gamma_i :=q_S/ q_i\le 1 \text{ with } i\in \{1,2\}.$
We are thus interested in the sign of
the quantity $g(y):= d_1 (y) d'_2 (y) - d_2(y)d'_1(y)$
which quickly leads to

$$
\frac{g(y) \delta}{y D_2} = \left(\frac{y}{Y_F}\right)^{\beta-1} 
\left[ (D_1 - D_2)(\beta+\frac{1}{\beta}-2 - \frac{\delta K}{D_2}) \right]
+\frac{\delta K}{D_2}\left(D_1-D_2\right)\; .
$$
Since $\beta>1$,
$(y/Y_F)^{\beta -1}$ is increasing in $y$.
Since $0<y\le Y_F$,
it suffices to verify that
$(\delta K)/D_2\ge (\beta+1/\beta-2 - (\delta K)/D_2)$,
which is naturally the case for any $\beta$,
to obtain that $g$ is non-negative on the interval.
\end{proof}

We will omit for now the symmetric case
and assume $P_1(y)<P_2(y)$ for $Y_L<y<Y_F$.
Recall now that $q_i>0$ for $i=1,2$.
Then
$P_i(Y_F) = 1/(q_i+q_S)>1$ for $i=1,2$.
Accordingly and by Lemma \ref{prop:concave_V},
there exists $Y_L<Y_1<Y_2<Y_F$ such that
	\begin{equation}
	\label{eq:Y1}
	\begin{array}{l}
	F(Y_1) = q_1 L(Y_1) + q_2 F(Y_1) + q_S S(Y_1)=S_1(Y_1)\;,\\
	F(Y_2) = q_2 L(Y_2) + q_1 F(Y_2) + q_S S(Y_2)=S_2(Y_2) \; .
	\end{array}
	\end{equation}

\begin{proposition} Assume $Y_L<y < Y_F$. Then solutions of Table \ref{tab:1}
are of three types:
\begin{enumerate}
\item[(a)] If $Y_L<y<Y_1$
the game
has three Nash Equilibria given by 
two pure strategies $(1,0)$ and $(1,0)$,
and one mixed strategy $(P_1(y),P_2(y))$.
\item[(b)] If $Y_1 \le 1<Y_2$, the game
has one Nash Equilibrium given by strategies $(1,0)$.
\item[(c)] If $Y_2 \le y<Y_F$, the game
has one Nash Equilibrium given by strategies $(1,1)$.
\end{enumerate}
\end{proposition}

\begin{proof}
For $Y_L<y<Y_1$, we have $P_1<P_2<1$.
Fix an arbitrary constant strategy
$(p_1, p_2)\in[0,1]^2$.
Considering the $\A$-expected payoffs,
agent one receives $L(y)$ at the end of the game
with a probability
	\begin{equation}
	\label{eq:a1}
	a_1 := p_1 \sum_{k\in \N^*} (1-p_1)^{k-1}(1-p_2)^k
	 = \frac{p_1 (1-p_2)}{p_1 + p_2 -p_1 p_2} \; .
	\end{equation}
Symmetrically, he receives $F(y)$, and agent two
receives $L(y)$ with probability
\begin{equation}
\label{eq:a2}
	a_2 := \dfrac{p_2 (1-p_1)}{p_1 + p_2 -p_1 p_2} \;,
\end{equation}
and they receives $(S_1(y), S_2(y))$ with probability
\begin{equation}
\label{eq:aS}
	a_S := \dfrac{p_1 p_2}{p_1 + p_2 -p_1 p_2} \;.
\end{equation}
The expected payoff of the game for agent one is given by
\begin{equation}
\label{eq:expected utility}
\begin{array}{ll}
E_1(p_1, p_2) &:= a_1 L(y)+a_2 F(y)+ a_S S_1(y)\\
& \ = (a_1+ a_S q_1)L(y) + (a_2 + a_S q_2)F(y)+a_S q_S S(y)\; .
\end{array}
\end{equation}
A similar expression $E_2$ is given for agent two.
Now fix $p_2$. Since $E_1$ is a continuous function of
both variables,
maximum of \ref{eq:expected utility} depends on
\begin{equation}
\label{eq:partial_derivative}
\frac{\partial E_1}{\partial p_1}(p_1,p_2)
= \frac{p_2 (L(y)-F(y) + p_2^2(S_1(y) - L(y))}{(p_1(1-p_2) + p_2)^2} \;.
\end{equation}
One can then see that
the sign of \ref{eq:partial_derivative}
is the sign of $(p_2-P_2)$.
A similar discrimination
for agent two implies
$P_1$.

\textbf{(a)}
If $Y_L<y<Y_1$, then
according to \ref{eq:V1} and \ref{eq:Y1},
$P_i< 1$ for $i=1,2$.
Three situations emerge:
\begin{enumerate}
\item[(i)] If $p_2>P_2$,
the optimal $p_1$ is $0$.
Then by \ref{eq:partial_derivative},
$E_2$ should not depend on $p_2$,
and the situation is stable for
any pair $(0,p_2)$ with $p_2$ in $(P_2,1]$.
\item[(ii)] If $p_2=P_2$,
$E_1$ is constant and $p_1$ can take any value.
If $p_1<P_1$, then by symmetry
$p_2$ should take value $1$, 
leading to case (i).
If $p_1=P_1$,
$E_2$ is constant and
either $p_2=P_2$, 
or we fall in case (i) or (iii).
The only possible equilibrium
is thus $(P_1, P_2)$.
\item [(iii)] If $p_2<P_2$,
$E_1$ is increasing with $p_1$ and
agent one shall play with probability $p_1=1> P_1$.
Therefore $p_2$ optimizes $E_2$ when being $0$,
and $E_1$ becomes independent of $p_1$.
Altogether,
situation stays unchanged if $p_1\in (P_1, 1]$
or if $p_1=0$.
Otherwise, if $p_1\le P_1$, we fall back
into cases (i) or (ii).
The equilibria here are
$(p_1, 0)$ with $p_1\in (P_1, 1]$,
and the trivial case $(0,0)$.
\end{enumerate}
Recalling 
constraint \ref{eq:strategy mixte},
we get rid of case $(0,0)$.
Coming back to the issue of the game
when $k$ goes to infinity in $(t,k,1)$,
three situations emerge
from the above calculation.
Two of them are
pure coordinated equilibriums,
of the type
$(a_1, a_2)=(1,0)$ 
or $(0,1)$,
which can be produced
by pure coordinated strategies
$(p_1, p_2)=(a_1, a_2)$,
settling the game in only one round.
The third one is
a mixed equilibrium given by
$
(p_1, p_2):=(P_1, P_2)
$.

	\textbf{(b)}
	According to \ref{eq:Y1},
	$S_1(Y_1)=F(Y_1)$.
	Following Lemma \ref{prop:concave_V},
	agent one prefers being a leader for $y\ge Y_1$
	and prefers a regulator intervention rather than
	the follower position, i.e. $S_1(y)\ge F(y)$.
	Thus $p_1=1$.
	For agent two,
	defering means receiving $F(y)>F(Y_1)$
	and exercising implies a regulation intervention.
	Since $y<Y_2$, $q_1 F(y) + q_2 L(y) + q_S S(y)<F(y)$ and
	defering is his best option: $p_2=0$.
	That means that on $(Y_1,Y_2)$, 
	the equilibrium strategy is $(p_1, p_2)=(1,0)$.
	
\textbf{(c)}
	On the interval $[Y_2, Y_F)$,
	the reasoning of (b) still applies for agent one
	by monotony, and $p_1=1$.
	The second agent can finally
	bear the same uncertainty if $y\ge Y_2$,
	and $p_2=1$.
	Here, $1\le P_1(y) \le P_2(y)$
	and both agents have greater
	expected payoff by letting
	the regulator intervene rather than
	being follower.
	Equilibrium exists
	when both agents exercise.
\end{proof}

Two reasons force us
to prefer the strategy $(P_1, P_2)$
on $(1,0)$ and $(0,1)$ in case (a). 
First, it is the only one
which extends naturally the symmetric case.
Second, it is the only trembling-hand equilibrium.
Considering $(P_1,P_2)$ in the interval (a),
$(a_1, a_2, a_S)$ follows
according to \ref{eq:a1}, \ref{eq:a2} and \ref{eq:aS},
	\begin{equation}
	\label{eq:outcomes_complete}
	(a_1, a_2, a_S)=
	\left(\frac{1-p_0}{2-p_0},
	\frac{1-p_0}{2-p_0}, 
	\frac{p_0}{2-p_0}\right)\;,
	\end{equation}
If we plug \ref{eq:outcomes_complete}
into \ref{eq:expected utility},
we obtain that the
payoff of respective agents
do not depend on $(q_1, q_2, q_S)$:
\begin{equation}
\label{eq:expected payoffs}
E_1(P_1, P_2)=E_2(P_1,P_2) = \frac{1-p_0}{2-p_0}(L(y)+F(y))+\frac{a_0}{2-p_0}S(y)\;.
\end{equation}
In the case $q_S>0$,
they are equal to $F$.
As notice in Grasselli \textit{et al.}\cite{grasselli2013},
we retrieve a mathematical expression of 
the \emph{rent equalization principle} of 
Fudenberg and Tirole \cite{fudenberg1985}:
agents are indifferent
between playing the game
and being the follower,
and time value of leadership vanishes with preemption.
In addition,
asymmetry $q_1\ge q_2$ does not
affect the payoffs
and the final outcome
of the game after decision
of the regulator
has the same probability as in 
the deregulated situation, see
Grasselli \textit{et al.} \cite{grasselli2013}.

\begin{figure}[htp]
\centering
\includegraphics[width=12cm, height=8cm]{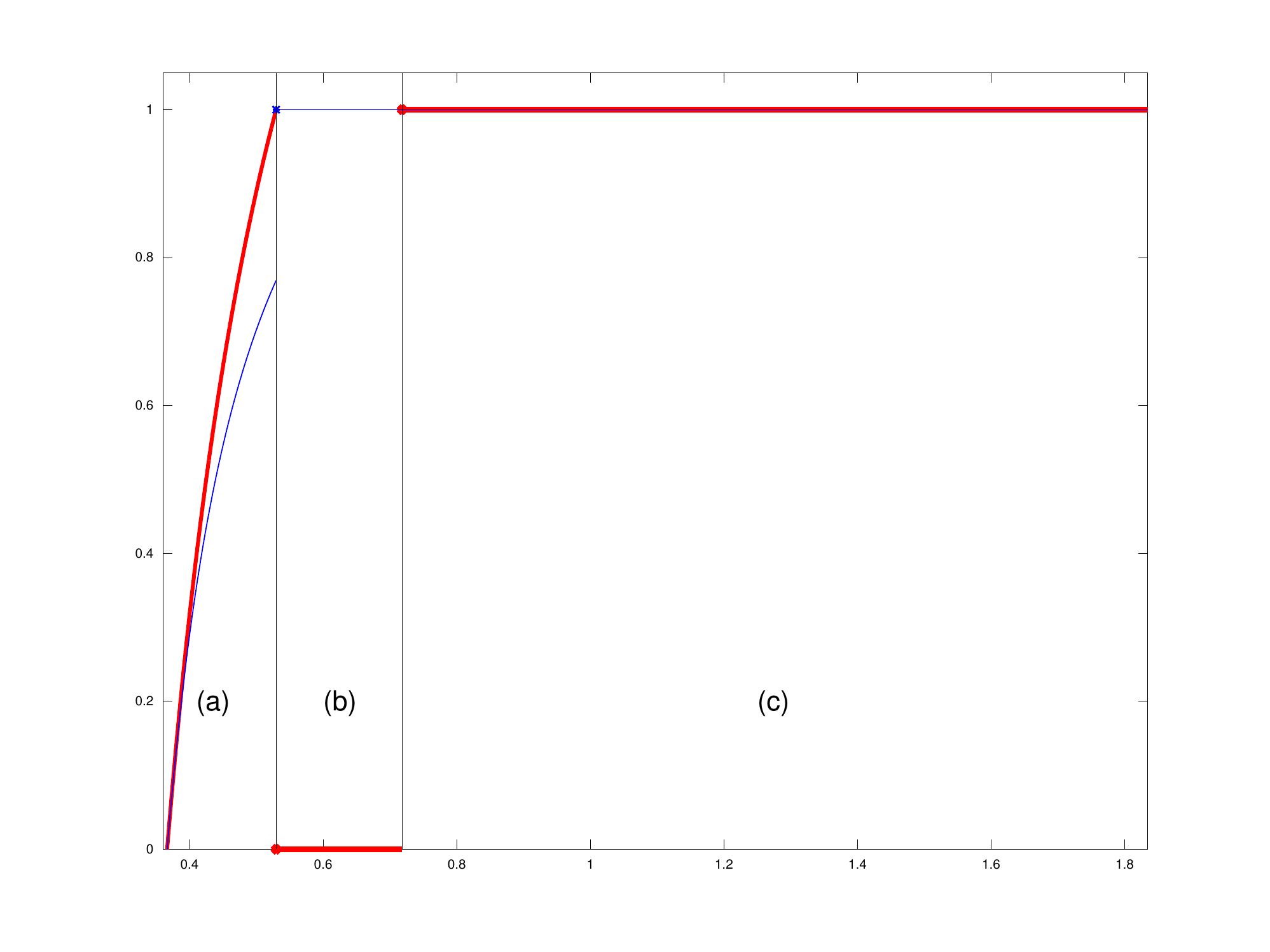}
\caption{Values of $p_1(y)$ (blue) and $p_2(y)$ (red) for $(q_1, q_2, q_S)=(0.5,0.2, 0.3)$.
Areas (a), (b) and (c) are separated by lines at
$Y_1=0.53$ and $Y_2=0.72$ on $[Y_L, Y_F]=[0.37,1.83]$.
Area (d) is at the left of the graph and (e) at the right of it.
Parameters set at 
$(K, \nu, \eta, \mu, \sigma, r, D_1, D_2)=(10,0.01,0.2,0.04,0.3,0.03,1,0.35)$.}
\label{fig:regions}
\end{figure}

\subsubsection{Endpoints and overall strategy}
We have to study junctions of areas (d) with (a), and (c) with (e).
The technical issue has been
settled in Thijssen \textit{et al.} \cite{thijssen2002}.

\begin{lemma}
\label{lem:YL}
Assume $y=Y_L$. Then both agents
have a probability $1/2$ to be leader or follower,
and receive $L(Y_L)=F(Y_L)$.
\end{lemma}

\begin{proof}
The junction of $[0,Y_L)$ with $[Y_L,Y_1]$ is a delicate point.
At the left of point $Y_L$,
no agent wants to invest.
We thus shall use the strategy
$G^i(Y_L)$ for both agents.
By right-continuity of this process,
both agents shall exercise 
with probability 1 at
point $Y_L$.
Each agent receives $R_i (y)$ which
takes value $L(y)=F(y)$ with probability
$q_1+q_2$ and $S(y)$ with probability $q_S$:
\begin{equation}
 \label{eq:payoff comparison}
 E_1(y)=E_2(y) = (q_1+q_2)F(y)+q_S S(y)\; .
\end{equation}
At the right side of point $Y_L$ however,
$P_i$ converges to $0$ when 
$y$ converges to $Y_L$, 
for $i=1,2$.
Therefore,
so do $(p_1, p_2)$ toward $(0,0)$.
We cannot reconcile $G^i_{\tau(Y_L)}(Y_L)$
with $(p_1(Y_L), p_2(Y_L))=(0,0)$
and shall compare the payoffs.
A short calculation provides
\begin{equation}
 \label{eq:limits bound}
 \lim_{y \downarrow Y_L} \frac{a_1(y)}{a_2(y)} =1  \And \lim_{y \downarrow Y_L} a_s(y) =0 \;.
\end{equation}
Therefore
at point $Y_L$,
$(a_1(Y_L), a_2(Y_L), a_S(Y_L))=(1/2, 1/2, 0)$.
It is clear at point $Y_L$
that the second option is better
for both agents.
\end{proof}	

\begin{remark}
	There is a continuity
	of behavior between (d) and (a)
	from the fact that
	regulatory intervention does not
	impact the outcome of the game,
	from the point of view of discrimination
	between the two agents:
	simultaneous investment is improbable
	at point $Y_L$,
	and probability $a_S(y)$
	is continuous and null at this point.
	Consequently,
	the local behavior of agents
	around $Y_L$ is similar 
	to the one given in
	Grasselli \textit{et al.} \cite{grasselli2013}.
	Altogether and observing Figure \ref{fig:regions},
	mixed strategies are right-continuous.
\end{remark}

Let us complete the strategic interaction
of the two agents by summarizing the previous results
in the following Theorem. As we will see in the next section
for singular cases, it extends Theorem 4 in \cite{grasselli2013}.

\begin{theorem}
\label{th:behavior}
Consider the strategic economic situation
presented in Section \ref{sec:game}, with
\begin{equation}
\min\{q_1, q_2, q_S\}> 0 \And q_0=0 \; .
\end{equation}
Then there exists a Markov sub-game perfect equilibrium
with strategies depending on the level of profits as follows:
\begin{enumerate}
\item[(i)] If $y<Y_L$, both agents wait for the profit level to rise and reach $Y_L$.
\item[(ii)] At $y=Y_L$, there is no simultaneous exercise and each agent
has an equal probability of emerging as a leader while the other
becomes a follower and waits until the profit level reaches $Y_F$.
\item [(iii)] If $Y_L<y<Y_1$, each agent choses a mixed strategy
consisting of exercising the option to invest with probability $P_i(y)$.
Their expected payoffs are equal, and the regulator intervenes on
the settlement of positions.
\item[(iv)] If $Y_1\le y <Y_2$, agent one exercises his option
and agent two becomes a follower and waits until $y$ reaches $Y_F$.
\item[(v)] If $Y_2\le y<Y_F$, both agents express their desire
to invest immediately, and the regulator is called.
If one agent is elected leader, the other one becomes
follower and waits until $y$ reaches $Y_F$.
\item[(vi)] If $Y_F \ge y$, both agents act as in (v), but
if a follower emerges, he invest immediately after the other agent. 
\end{enumerate}
\end{theorem}

Some comments are in order.
First, we shall emphasize that
the regulator theoretically intervenes
for any situation where $y\ge Y_L$.
However, as explained at the beginning of 
this section, its intervention
becomes mostly irrelevant in continuous time
when agents do not act simultaneously.
Its impact is unavoidable to settle
the final payoff after the game,
but its influence on agents' strategies
boils down to the interval $(Y_L, Y_2]$.

At point $Y_1$,
there is a strong 
discontinuity
in the optimal behavior
of both agents.
For $y<Y_1$,
the mixed strategy
used by agent two
tends toward a pure strategy
with systematic investment.
However at the point itself,
the second agent defers
investment and becomes
the follower.
It follows from the fact that
agent one is indifferent between
being follower or letting
the regulator decide the outcome at this point.
He suddenly seeks, without hesitation, for the
leader's position,
creating a discontinuity
in his behavior.
The same happens for agent two
at $Y_2$, creating another discontinuity
in the strategy of the latter.

\section{Singular cases}
\label{sec:singular}

The proposed framework encompasses
in a natural way the two main competition
situations encountered in the literature, 
namely the Cournot competition and the 
Stackelberg competition.
By introducing minor changes 
in the regulatory intervention,
it is also possible to
represent the situation of
Stackelberg leadership advantage.
Finally,
our setting allows to
study a new and weaker type of advantage
we call \emph{weak Stackelberg leadership}.

\subsection{The Cournot and Stackelberg competitions}
\label{sec:singular_cournot}

The Cournot duopoly refers to
a situation when both players 
adjust quantities simultaneously.
This is the framework of
Grasseli \textit{et al.}
\cite{grasselli2013}, 
involving the payoff $S(y)$ if
agents invest at the same time.
Recalling Definition \ref{def:regulator},
this framework corresponds to 
\begin{equation}
\label{eq:cournot}
(q_1, q_2, q_S) = (0,0,1) \; .
\end{equation}
We notice that
agents then become symmetrical.
This appears from
\ref{eq:V1} which
implies $P_1(y)=P_2(y)=p_0(y)$.
Additionally, $p_0(y)\in (0,1)$
if $y\in (Y_L, Y_F)$,
and is increasing on this interval
according to Lemma \ref{prop:thresholds}.

The Stackelberg competition
refers to a situation
where competitors adjust the
level $D_{Q(t)}$ sequentially.
In a preemptive game, this implies
that in the case of simultaneous
investment, one agent is elected leader
and the other one becomes follower.
This setting implies an exogenous
randomization procedure, 
which by symmetry is given by
a flipping of a fair coin.
The procedure is described as such in
Grenadier \cite{grenadier1996},
Weeds \cite{weeds2002},
Tsekeros \cite{tsekrekos2003}
or Paxson and Pinto \cite{paxson2005}.
Recalling Definition \ref{def:regulator},
the setting is retrieved by fixing
\begin{equation}
\label{eq:stackelberg}
(q_1,q_2, q_S)=(1/2, 1/2, 0)\; .
\end{equation}
The implication in our context
is the following:
by symmetry,
$P_1(y) = P_2(y) = 2$.
Therefore, the interval
$(Y_L, Y_2)$ reduces to nothing,
i.e., $Y_L=Y_1=Y_2$,
and the strategical behavior boils down
to (i), (v) and (vi)
in Theorem \ref{th:behavior}.

\begin{remark}
\label{rem:stackelberg}
Notice that any 
combination $q_2=1-q_1 \in (0,1)$
provides the same result as in \ref{eq:stackelberg}, 
as $P_i(y) = 1/q_i>1$ for $i=1,2$.
The strategic behavior is unchanged
with an unfair coin flipping.
This is foreseeable as
$q_i L(y)+ (1-q_i)F(y)>F(y)$
on $(Y_L, Y_F)$.
This will also hold for a convex
combination, i.e.,
for risk-averse agents.
\end{remark}

\begin{remark}
\label{rem:continuous}
Assume now symmetry in the
initial framework, i.e.,
$q:=q_1=q_2 \in (0, 1/2)$.
We have $Y_S:=Y_1=Y_2$, 
and the region (b) reduces to nothing.
By recalling \ref{eq:Y1},
we straightly obtain
\begin{equation}
\lim\limits_{q\uparrow 1/2}Y_S = Y_L \ \And \
\lim\limits_{q\downarrow 0}Y_S = Y_F \; .
\end{equation}
Therefore, 
the regulation intervention
encompasses in a continuous 
manner
the two usual types of games described
above.
\end{remark}

\subsection{Stackelberg leadership}
\label{sec:stackelberg leadership}

This economic situation
represents an asymmetrical competition
where the roles of leader and follower
are predetermined exogenously.
It can be justified as in Bensoussan \textit{et al.}
\cite{bensoussan2010} by regulatory
or competitive advantage.
However Definition \ref{def:regulator}
does not allow to retrieve this case directly.
Instead, we can extend it by
conditioning the probability quartet $\{q_0, q_1, q_2, q_S\}$.
In this situation,
the probability $\P^+$ depends 
on $\F$-adapted events in the following manner:
\begin{equation}
\label{eq:leadership}
\P^+(\alpha_0 | t < \tau_1)=1 \And \P^+ (\alpha_1 | t\ge \tau_1 )=1 \; .
\end{equation}
This means that no
investment is allowed until agent one
decides to invest, which leads
automatically to the leader position.
The strategical interaction is then pretty different from
the endogenous attribution of roles.
See Grasselli and \textit{al.}\cite{grasselli2013}
for a comparison of this situation
with the Cournot game.

\subsection{The weak Stackelberg advantage}
\label{sec:weak}

We propose here a different type of
competitive situation.
Consider the investment timing problem
in a duopoly where agent one has, 
as in the Stackelberg leadership,
a significant advantage due to
exogenous qualities.
We assume that for particular reasons,
the regulator only allows one
agent at a time to invest in
the shared opportunity.
The advantage of agent one thus
translates into a preference
of the regulator, but only in the case of
simultaneous move of the two agents.
That means that agent two
can still, in theory, preempt agent one.
This situation can be covered by simply setting
\begin{equation}
\label{eq:weak stackelberg}
(q_1,q_2, q_S)=(1, 0, 0)\; .
\end{equation}
In this setting,
the results of Section \ref{sec:nash}
apply without loss of generality.

\begin{proposition}
\label{prop:weak stackelberg}
Assume \ref{eq:weak stackelberg}. 
Then the optimal strategy for
agents one and two is given by
$(G^1 (Y_L), G^2(Y_F))$.
\end{proposition}

\begin{proof}
It suffices to
consider the game of Table \ref{tab:1}
on the interval $[Y_L, Y_F)$.
The Nash equilibrium is given
by $(p_1(y), p_2(y))=(1,0)$ for 
any $y\in [Y_L, Y_F]$.
Corollary \ref{cor:d and e} applies
for other values.
\end{proof}

\begin{remark}
As in Remark \ref{rem:continuous},
we can observe that this situation
corresponds to strategy (iv) of Theorem \ref{th:behavior}
expanded to the interval $[Y_L, Y_F)$.
This situation can be obtained continuously
with $q_1$ converging to one.
\end{remark}

Following Grasselli \textit{et al.} \cite{grasselli2013},
we compare the advantage given
by such an asymmetry to the usual
Cournot game.
In the latter,
the positive difference
between the Stackelberg leadership
and the Cournot game
provides a financial advantage
which is called a priority option.
By similarity, we call
the marginal advantage of the weak Stackelberg
advantage on the Cournot game
a \textit{preference option}.

\begin{corollary}
\label{cor:preference option}
Let us assume $q_2=q_0=0$.
Let us denote $E_1^{(q_1,q_S)}(y)$
the expected payoff of agent one 
following from Theorem \ref{th:behavior}
when $(q_1, q_S)$ is given, for a level $y\in\R_+$.
Then the preference option value is given by
$\pi^0(y):=E_1^{(1,0)}(y) - E_1^{(0,1)}(y)$ 
for all $y\in \R_+$.
Its value is equal to
\begin{equation}
\label{eq:preference option}
\pi^0(y) = (L(y)-F(y))^+ \quad \forall y\in \R_+ \; .
\end{equation}
\end{corollary}

\begin{proof}
The proof is straightforward and follows
from \ref{eq:expected payoffs}, where
$E_1^{(0,1)}(y)=F(y)$ for $y\in [Y_L, Y_F]$.
Following Proposition \ref{prop:weak stackelberg},
agent one shall invest at $\tau(Y_L)$ if $y\le Y_L$
at first. In this condition,
his payoff is $L(Y_L)=F(Y_L)$, which provides
\ref{eq:preference option}.
\end{proof}

This option gives
an advantage to its
owner, agent one, without
penalizing the other agent,
who can always expect the
payoff $F$.
A comparison with the priority option is given
in Figure \ref{fig:options}.

\begin{figure}[ht]
\centering
\includegraphics[width=12cm, height=8cm]{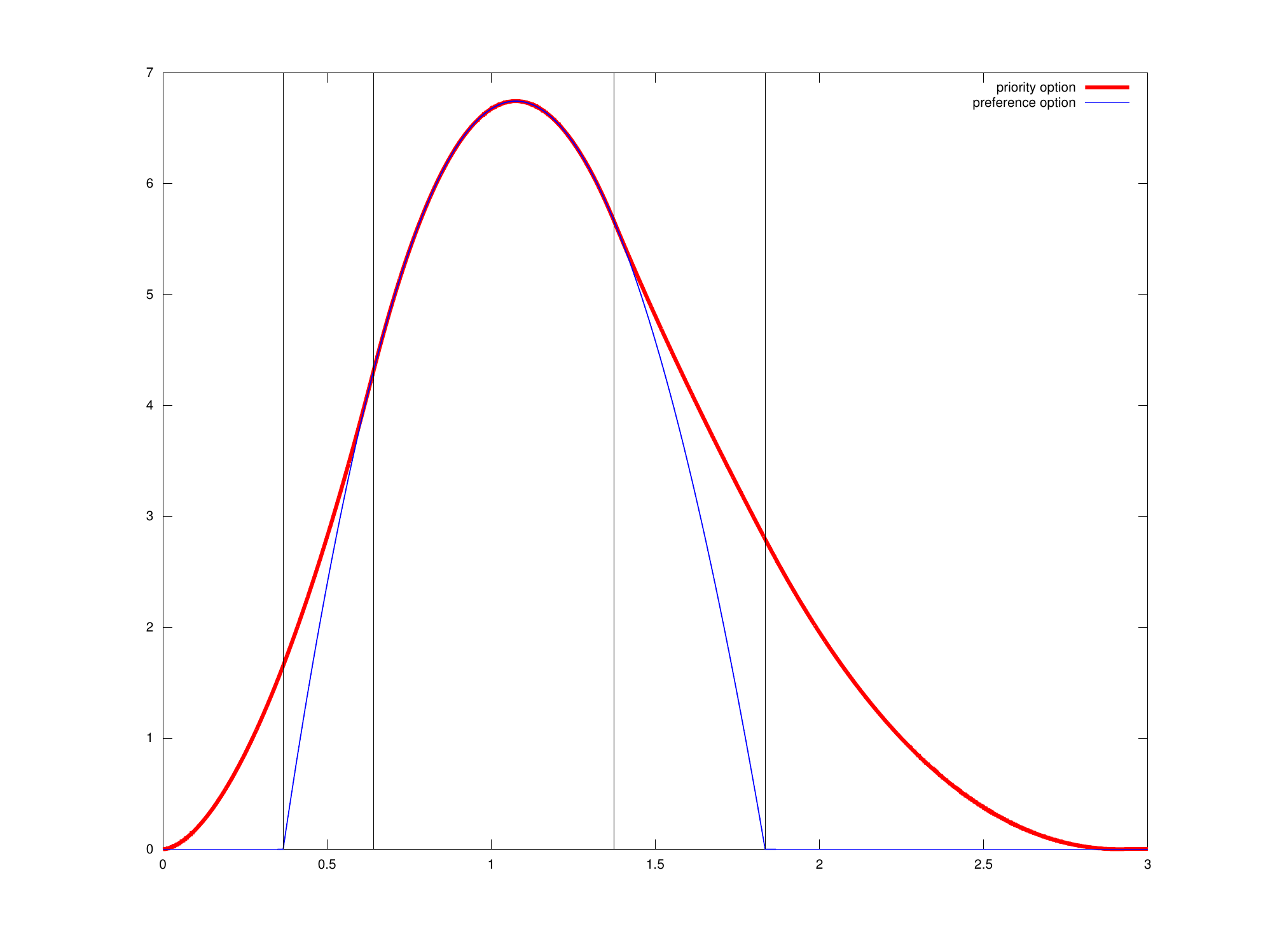}
\caption{Priority option value (red) and
Preference option value (blue) in function of $y$.
Vertical lines at $Y_L=0.37$, $Y_1=0.64$, $Y_2=1.37$ and $Y_F=1.83$.
Option values are equal on $[Y_1, Y_2]$.
Same parameters as in Figure \ref{fig:regions}.}
\label{fig:options}
\end{figure}

A very similar situation to the
weak Stackelberg advantage can be observed
if we take $q_2=0$ but $q_S>0$.
The economic situation however loses
a part of its meaning.
It would convey the situation where
agent two is accepted as investor
only if he shares the market with agent one
or become a follower.
In this setting we can apply also
the results of Section \ref{sec:nash}.
We observe from
definition \ref{eq:Y1}
that in this case,
$Y_2=Y_F$
and $Y_1$
verifies
$F(Y_1) = q_1 L(Y_1) + (1-q_1)S(Y_1)$.
The consequence is straightforward:
the interval $[Y_1, Y_2)$
expands to $[Y_1, Y_F)$.
The fact that $Y_1>Y_L$
for $q_1<1$ has
also a specific implication.

In that case, the equilibrium
$(1,0)$ is more relevant
than $(P_1(y), P_2(y))$
on $[Y_L, Y_1)$.
Indeed,
if agent one invests systematically
on that interval
then agent two has no chance 
of being the leader
since $q_2=0$.
Thus $p_2=0$ and the payoffs become
\begin{equation}
\label{eq:expected_payoff2}
(E_1(y), E_2(y))=\left( L(y), F(y) \right) \text{ for }y\in [Y_L, Y_1)\; .
\end{equation}
In opposition of the
trembling-hand equilibrium,
this pure strategy
can be well figured by
being called
\emph{steady-hand}
Comparing \ref{eq:expected payoffs}
to \ref{eq:expected_payoff2},
agent one shall use the pure strategy,
whereas agent two is indifferent between
the mixed and the pure strategy.
Setting $q_2=0$ thus provides
a preference option to agent one.

\section{Aversion for confrontation}
\label{sec:aversion}

The introduction of the regulator
rendering the market incomplete,
utility indifference pricing can be invoked
to evaluate payoffs.
Although its action
is purely random in the presented model,
the regulator represents a third
actor in the game.
We study here the the impact
of risk aversion
on the coordination game only,
which can be seen as an original
dimension we denote
\emph{aversion for confrontation}.

\subsection{Risk aversion in complete market}
We keep the market model of 
Section \ref{sec:game}.
However,
we now endow each agent
with the same CARA utility function
\begin{equation}
U(x)=- \exp (- \gamma x)
\end{equation}
where $\gamma>0$ is the risk-aversion of agents.
The function $U$ is strictly concave and
is chosen to avoid
dependence on initial wealth.
Recalling Remark \ref{rem:incomplete},
the market without regulator is
complete and
free of arbitrage,
so that both agents still
price the leader,
the follower and 
sharing positions with
the unique risk-neutral probability $\Q$.
Thus,
agents compare the utility
of the different investment option market prices, 
denoted
$l(y):=U(L(y))$, 
$f(y):=U(F(y))$
and $s(y):=U(S(y))$.
The $s_i$, $i=1,2$
are defined the same way.
When needed,
variables will be indexed
with $\gamma$ to make 
the dependence explicit.
The definition of regulation is updated.

\begin{definition}
\label{def:regulator2}
Fix $t$ and $Y_t=y$. For $i=1,2$, 
if $j=3-i$ is the index of the opponent
and $\tau_j$ his time of investment, 
then agent $i$ receives utility
\begin{equation}
\label{eq:regulator2}
r_i(t,y):=\left\{
\begin{array}{ll}
0 & \text{if }\alpha=\alpha_0\\
l(y)\ind{t\le \tau_j}+f(y)\ind{t>\tau_j} & \text{if }\alpha=\alpha_i\\
f(y)\ind{t=\tau_j} & \text{if }\alpha=\alpha_j\\
l(y)\ind{t< \tau_j}+s(y)\ind{t=\tau_j}+f(y)\ind{t>\tau_j} & \text{if }\alpha=\alpha_S\\
\end{array}
\right. \; .
\end{equation}
\end{definition}

From Definition  \ref{def:regulator2},
it appears that by
monotonicity of $U$,
the game is strictly the same as in Section \ref{sec:game},
apart from payoffs.
Both agents defer for $y<Y_L$ and
both act immediately for $y\ge Y_F$.
The incomplete market setting is now
handled with a utility maximization criterion.
Each agent uses now a mixed strategy $p^{\gamma}_i$,
in order to maximize an expected utility 
slightly different from \ref{eq:expected utility} on
$(Y_L, Y_F)$:
\begin{equation}
\label{eq:expected_utility_incomplete}
E^{\gamma}_1(y)=(a^{\gamma}_1+a^{\gamma}_s q_1) l(y) + (a^{\gamma}_2+a^{\gamma}_S q_2) f(y)+ a^{\gamma}_s q_S s(y) 
\end{equation}
Results of Section \ref{sec:nash}
thus hold by
changing $(L, F, S)$ for $(l, f, s)$.
It follows that in that case
\begin{equation}
\label{eq:pi_incomplete}
P_{i, \gamma}(y)  = \frac{l(y)-f(y)}{q_i(l(y)-f(y))-q_{S} (l(y)-s(y))} \text{ with } i\in \{1,2\}
\end{equation}
play the central role
and that
optimal strategic interaction
of agents
can be characterized by the value of $y$,
or equivalently on the interval $[Y_F, Y_L]$ by
the values of $P_{1, \gamma}(y)$ and $P_{2,\gamma}(y)$.
The question we address
in this section is how 
risk-aversion influences the
different strategic interactions.

\subsection{Influence of $\gamma$ on strategies} 

First,
aversion for confrontation
is expressed through diminishing
probability of intervention with $\gamma$.

\begin{proposition}
\label{prop:risk on probability}
Assume $q_S=1$, and $y\in (Y_L, Y_S)$.
Denote $p_\gamma(y):=P_{1,\gamma}(y)=P_{2,\gamma}(y)$.
Then $p_{\gamma}\le p_0$ and furthermore,
\begin{equation}
\label{eq:limit risk aversion}
\lim\limits_{\gamma\downarrow 0}p_{\gamma}=p_0 \And
\lim\limits_{\gamma \uparrow \infty} p_\gamma=0\; .
\end{equation}
\end{proposition}

\begin{proof}
According to 
\ref{eq:pi_incomplete} with $q_S=1$,
$p_{\gamma}(y)\in (0,1)$ on $(Y_L, Y_S)$.
From \ref{eq:pi_incomplete} we get
\begin{equation}
p_\gamma(y)=\frac{l(y)-f(y)}{l(y)-s(y)}
= \frac{-e^{\gamma L(y)}+e^{\gamma F(y)}}{-e^{\gamma L(y)}+e^{\gamma S(y)}}
= \frac{e^{\gamma (L-F(y))}-1}{e^{\gamma (L-S(y))}-1}\; .
\end{equation}
Since 
$u(x) := -1-U(-x) = e^{\gamma x}-1$ is
a positive strictly convex function on $\R_+$
with $u(0)=0$,
we have that
\begin{equation}
\label{eq:probability_gamma}
p_\gamma (y)=\frac{u(L(y)-F(y))}{u(L(y)-S(y))}< \frac{L(y)-F(y)}{L(y)-S(y)} =: p_0(y)\; .
\end{equation}
For $\gamma$ going to zero,
we apply l'H\^{o}pital's rule
to obtain that 
$\lim_{\gamma \downarrow 0} p_\gamma = p_0$.
The other limit follows from
expliciting the utility function:
\begin{equation}
\label{eq: p limit gamma}
\lim\limits_{\gamma \uparrow \infty} p_\gamma (y)=\lim\limits_{\gamma \uparrow \infty} e^{-\gamma (F(y)-S(y))}=0\; .
\end{equation}
\end{proof}

\begin{remark}
\label{rem:p gamma}
Notice that there is no
uniform convergence since
$p_0$ is continuous and
$p_0(Y_F)=1$.
The above convergence holds for all $y\in [Y_L, Y_F)$.
It is clear from \ref{eq:probability_gamma}
that $p_\gamma$ is monotonous
in $\gamma$ on $\R_+^*$.
Then according to \ref{eq: p limit gamma},
it is convex decreasing with $\gamma$.
\end{remark}

For the general case $q_S<1$,
the above result still holds, but
on a reduced interval.
Nevertheless, this interval depends on $\gamma$.

\begin{proposition}
\label{prop:modifications}
Assume $\min\{q_1, q_2, q_S\}>0$.
Let $Y_{1,\gamma}\in [Y_L, Y_F]$ be such that
$P_{2,\gamma}(Y_{1,\gamma})=1$, and 
$Y_{2,\gamma}\in [Y_L, Y_F]$ such that
$P_{1,\gamma}(Y_{2,\gamma})=1$.
Then for $i=1,2$,
$Y_{i,\gamma}$ is increasing in $\gamma$
and
\begin{equation}
\label{eq: y gamma}
\lim\limits_{\gamma \uparrow \infty} Y_{i,\gamma} = 
Y_F\; .
\end{equation}
\end{proposition}

\begin{proof}
Consider $i\in\{1,2\}$.
First notice that along Lemma \ref{prop:thresholds},
$Y_{i,\gamma}$ is uniquely defined on
the designated interval.
Following Remark \ref{rem:p gamma},
$P_{i,\gamma}$ is a concave non-decreasing
functions of $p_\gamma$.
It is then a decreasing function of $\gamma$.
Since $P_{i,\gamma}(y)$ is decreasing with $\gamma$,
$Y_{i,\gamma}$ is an increasing function of $\gamma$:
the region $(Y_L, Y_{1,\gamma})$ spreads on the right with $\gamma$.
Adapting \ref{eq:Y1} to the present values,
$Y_{i, \gamma}$ shall verify:

$$
q_i (1-e^{\gamma (L(Y_{i, \gamma})-F(Y_{i, \gamma}))}) + q_S (1 - e^{\gamma (S(Y_{i, \gamma})-F(Y_{i, \gamma}))}) = 0
$$
and when $\gamma$ goes to $\infty$,
following \ref{eq: p limit gamma},
we need $L(Y_{i, \gamma})-F(Y_{i, \gamma})$ to go to 0,
so that $Y_{i, \gamma}$ tends toward $Y_F$.
\end{proof}

With risk aversion,
the region (a) takes more importance
and the competitive advantage of
agent one decreases with $\gamma$.
Figure \ref{fig:areas_aversion}
resumes the evolution.

\begin{figure}[ht]
\centering
\includegraphics[width=12cm, height=8cm]{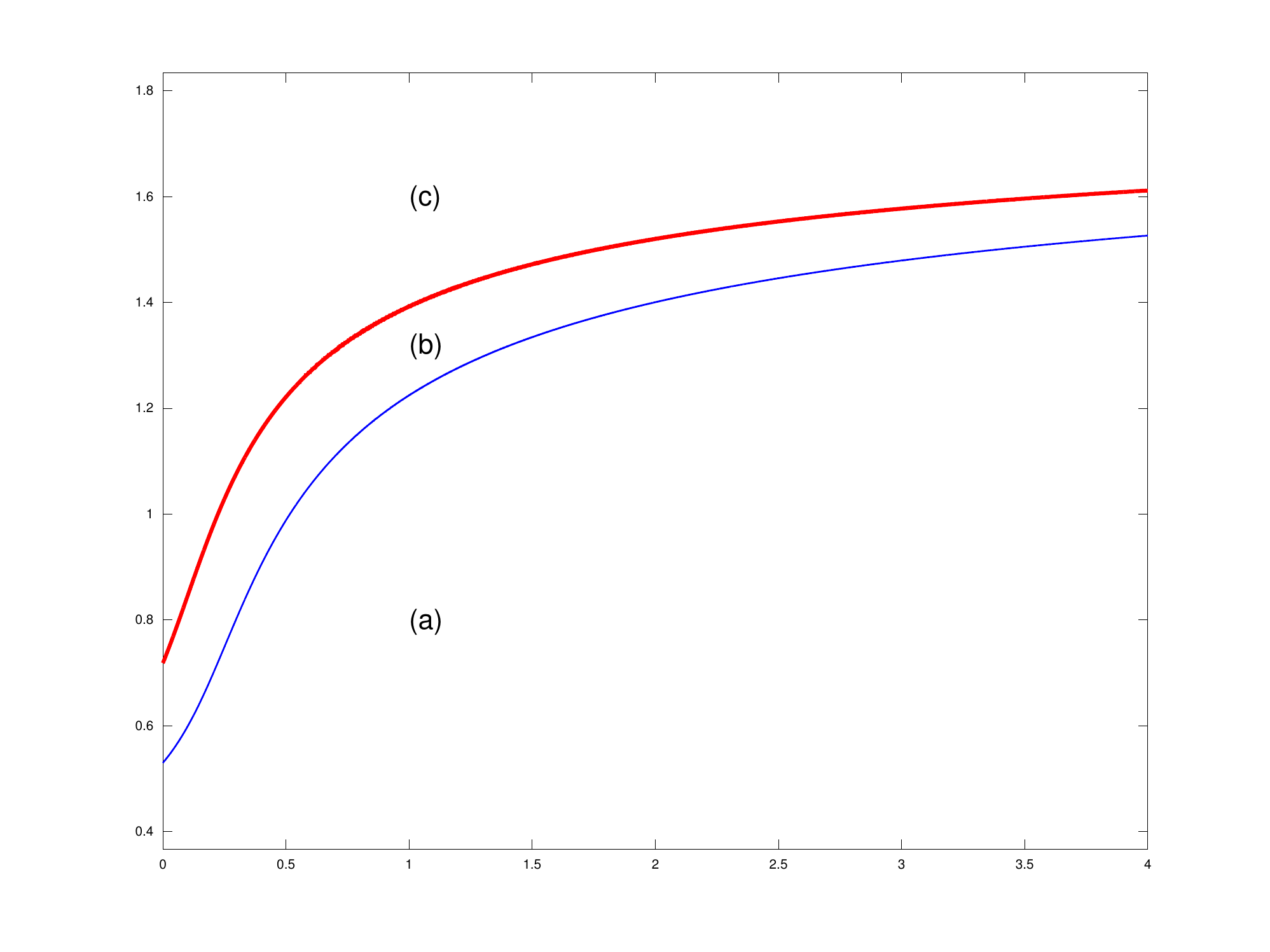}
\caption{
Values of $Y_1$ (blue) and $Y_2$ (red) as a function of risk aversion $\gamma$.
Y-axis limited to $[Y_L, Y_F]=[0.37, 1.83]$.
Limit values of $(Y_1,Y_2)$ for $\gamma$ going to 0
corresponds to $(0.53, 0.72)$ of Figure \ref{fig:regions}.
Same parameters as previous figures.}
\label{fig:areas_aversion}
\end{figure}

The interval $(Y_L, Y_{1,\gamma})$
describing the case (a) is the 
more relevant in terms of strategies,
as the only one to involve mixed strategies.
Propositions \ref{prop:risk on probability} and \ref{prop:modifications}
imply the following result on this interval.

\begin{corollary}
\label{eq: outcomes with gamma}
Assume $y\in [Y_L, Y_{1,\gamma})$.
Let 
$(a_1^{\gamma},a_2^{\gamma},a_S^{\gamma})$
be defined by \ref{eq:a1}, \ref{eq:a2}
and \ref{eq:aS} where $(p_1, p_2)$ is replaced
by $(P_{1,\gamma}, P_{2,\gamma})$. Then
\begin{equation}
\label{eq:lim de a}
\lim\limits_{\gamma\uparrow\infty}a_S =0 \And
\lim\limits_{\gamma\uparrow\infty} \frac{a_1^\gamma}{a_2^\gamma} = 1 \; .
\end{equation}
\end{corollary}

\begin{proof}
Denote $p_i^\gamma:=P_{i, \gamma}$ and
$p_\gamma:= (l(y)- f(y))/(l(y)-s(y))$ the
equivalent of $p_0$ in the risk averse setting.
From \ref{eq:aS}, $a_S^\gamma$ is a decreasing function
of both $p_1^\gamma$ and $p_2^\gamma$, 
and the first limit is a straightforward
consequence of \ref{eq:limit risk aversion}.

Plugging \ref{eq:pi_incomplete}
into $a_i^\gamma$,
we obtain
\begin{equation}
\label{eq:a_i}
a^{\gamma}_i = \frac{p_i^\gamma (1-p_j^\gamma)}{p_1^\gamma + p_2^\gamma - p_1^\gamma p_2^\gamma}, 
\quad i\ne j \in \{1,2\}^2 \; ,
\end{equation} 
and differentiating in $p_\gamma$,
we obtain that $a^{\gamma}_i$
is increasing in $\gamma$.
It also follows that
\begin{equation}
\label{eq:rapport de force1}
\frac{a^{\gamma}_1}{a^{\gamma}_2} 
= \frac{p^{\gamma}_1 - p^{\gamma}_1 p^{\gamma}_2}{p^{\gamma}_2 - p^{\gamma}_1 p^{\gamma}_2} 
= \frac{q_S -(1-q_2) p_\gamma}{q_S -(1-q_1) p_\gamma}\le 1 \; .
\end{equation}
The second limit of \ref{eq:lim de a} 
follows immediately from Proposition \ref{prop:risk on probability}.
\end{proof}

\begin{remark}
We can easily assert as in Proposition \ref{prop:risk on probability}
that $a_S^\gamma<a_S$ for any $\gamma>0$, where $a_S$
is the probability of simultaneous action in the game
for risk-neutral agents.
Equivalently,
\begin{equation}
\frac{a_1}{a_2}
=\frac{F(y)-S_1(y)}{F(y)-S_2(y)}< \frac{a_1^\gamma}{a_2^\gamma}, \quad \forall \gamma>0\; .
\end{equation}
\end{remark}

The above results can be interpreted as follows.
The parameter $\gamma$ is an aversion
for the uncertainty following from the
coordination game and the regulatory intervention.
As expected,
the higher this parameter, the lower
the probability to act $P_{i, \gamma}$ in the game.
However the game being infinitely repeated until
action of at least one agent,
for values of $P_{i, \gamma}$ lower than one,
the game is extended to a bigger interval $[Y_L, Y_{1,\gamma})$
with $\gamma$.
Then, it naturally reduces simultaneity of investment,
but tends also to an even chance of becoming
a leader for both agents.
There is thus a consequence
to a high risk-aversion $\gamma$:
agents synchronize to avoid playing the game
and the regulator decision.
In some sense,
the behavior in a Stackelberg competition
is the limit with risk-aversion of
the behavior of competitors in a Cournot game.

How does $\gamma$ impact
the outcome of the
game?
As above
since we are in complete market,
the risk aversion makes no relevant
modification to the cases
where the coordination game is not played.
On the interval $(Y_L, Y_{1,\gamma})$
then, we need to compare 
the expected values of
options $L, F$ and $S$
to homogeneously quantities.

\begin{proposition}
\label{def:indifference prices}
For $y\in [Y_L, Y_{1, \gamma})$,
let $E_1^{\gamma}(y)$ be the expected utility
of agent one in the coordination game
when strategies $(P_{1,\gamma}(y), P_{2,\gamma}(y))$ are played.
Then the indifference value of the game for
agent one is given by
\begin{equation}
\label{eq:indifference price}
e_{1,\gamma}(y):= U^{-1}(E_1^{\gamma}(y))
=-\frac{1}{\gamma}\log\left( -a^{\gamma}_1 l(y)-a^{\gamma}_2 f(y)-a^{\gamma}_S s_i(y) \right) \; .
\end{equation}
We define $e_{2,\gamma}$ similarly.
Assume now $q_S>0$. 
For $y\in [Y_L, Y_{1, \gamma})$, we have for $i=1,2$ that
\begin{equation}
\label{eq:indifference prices}
e_{i,\gamma}(y)= E_i (y)=F(y) \quad \forall \gamma>0 \; ,
\end{equation}
with $E_{i}(y)$ defined in \ref{eq:expected payoffs}.
\end{proposition}

\begin{proof}
Using \ref{eq:a_i},
we can proceed as in the end of Section \ref{sec:nash} to retrieve
\begin{equation}
\label{eq:sum of prices}
a_{1,\gamma}l(y)+ a_{2,\gamma}f(y)+a_{S,\gamma}s_i(y) = f(y), \quad i=1,2 \; ,
\end{equation} 
and \ref{eq:indifference prices} follows from the definition of $e_{i,\gamma}$.
\end{proof}

Risk-aversion in complete market has thus
the interesting property to keep true the
rent equalization principle of Fudenberg and Tirole \cite{fudenberg1985}:
agents adapt their strategies to be indifferent
between playing the game or
obtaining the follower's position.

\section{Criticism and extensions}
\label{sec:conclusion}

If real option games model
are to be applied,
the role of a regulator
shall be introduced in order
to suit realistic situations of competition.
Indeed regulators often intervene
for important projects in energy,
territorial acquisition and
highly sensitive products such as drugs.
Despite its simplicity and
its idealization,
the present model has still something to say.
One can see it as
an archetype model,
unifying mathematically
the Stackelberg and the Cournot
competition frameworks
and leading to simple formulas.
The model we proposed
is a first attempt,
and could be improved on several grounds.
Three main objections can be raised.

As observed in Remark \ref{rem:reduction} of Section \ref{sec:nash},
the extention of strategies in \ref{sec:strategies}
in the fashion of Fudenberg and Tirole \cite{fudenberg1985}
has unrealistic and strongly constrained
implications due to continuous time and perfect information.
This is emphasized here in the interpretation that the regulator can
only intervene when agents simultaneously invest.
A prioritize objective
would be to propose a new setting
for the standard game
conveying a dynamical dimension to
strategies.

A realistic but involved
complication is the influence
of explicit parameters on
the law $\P^+$,
parameters on which agents
have some control.
The value of being preferred,
introduced as a financial option in subsection \ref{sec:weak},
would then provide a price
to a competitive advantage and
to side efforts to satisfy
non-financial criteria.
Eventually, a game
with the regulator as a third player
appears as a natural track of inquiry.

Finally,
the introduction of
risk-aversion should not
be restrained to the coordination game.
This issue is already tackled in 
Bensoussan \textit{et al.} \cite{bensoussan2010}, 
and Grasselli \textit{et al.} \cite{grasselli2013}
for the incomplete market setting.
But as a fundamentally different source
of risk, the coordination game
and the regulator's decision shall
be evaluated with a different risk-aversion parameter,
as we proposed in Section \ref{sec:aversion}.
An analysis of asymmetrical risk aversion
as in Appendix C of Grasselli \textit{et al.} \cite{grasselli2013}
shall naturally be undertaken.

\end{document}